\documentclass[10pt]{llncs}
\usepackage{array}
\usepackage{float}
\usepackage[dvipdfmx, usenames]{color}
\usepackage[dvipdfmx,hiresbb]{graphicx}
\usepackage{amsmath, amssymb}
\usepackage{algpseudocode,algorithm}
\usepackage{cases}

\newcommand{\argmax}{\operatornamewithlimits{\mathrm{arg\,max}}}

\title{Derandomization for $k$-submodular maximization}
\author{Hiroki Oshima}
\institute{Department of Mathematical Informatics,\\Graduate School of Information Science and Technology, The University of Tokyo}
\date{\today}

\begin{document}
\maketitle
\begin{abstract}
Submodularity is one of the most important property of combinatorial optimization, and $k$-submodularity is a generalization of submodularity. Maximization of a $k$-submodular function is NP-hard, and approximation algorithm has been studied.
For monotone $k$-submodular functions, [Iwata, Tanigawa, and Yoshida 2016] gave $k/(2k-1)$-approximation algorithm. In this paper, we give a deterministic algorithm by derandomizing that algorithm. Our algorithm is $k/(2k-1)$-approximation and runs in polynomial time. 
\end{abstract}
\section{Introduction}
A set function $f:2^V \to \mathbb{R}$ is submodular if, for any $A, B \subseteq V$,
$f(A)+f(B)  \geq f(A \cup B)+f(A \cap B)$. Submodularity is one of the most important properties of combinatorial optimization.  The rank functions of matroids and cut capacity functions of networks are submodular. Submodular functions can be seen as discrete version of convex functions.

For submodular function minimization, \cite{grotschel1981ellipsoid} showed the first polynomial-time algorithm. The combinatorial strongly polynomial algorithms were shown by \cite{iwata2001combinatorial} and \cite{schrijver2000combinatorial}. On the other hand, submodular function maximization is NP-hard and we can only use approximation algorithms.  Let an input function for maximization be $f$, a maximizer of $f$ be $S^*$, and an output of an algorithm be $S$. The approximation ratio of the algorithm is defined as $f(S)/f(S^*)$ for deterministic algorithms and $\mathbb{E}[f(S)]/f(S^*)$ for randomized algorithms. A randomized version of Double Greedy algorithms in \cite{buchbinder2015tight} achieves $1/2$-approximation. \cite{feige2011maximizing} showed $(1/2 +\epsilon )$-approximation requires exponential time value oracle queries. This implies that, Double Greedy algorithm is one of the best algorithms in terms of the approximation ratio. \cite{buchbinder2016deterministic} showed a derandomized version of randomized Double Greedy  algorithm, and their algorithm achieves $1/2$-approximation.

$k$-submodularity is an extension of submodularity.  $k$-submodular function is defined as below.
\begin{definition}
\label{defnksubmo}
Let $(k+1)^{V}:=\{(X_1,...,X_k) \mid X_i\subseteq V \ (i=1,...,k) ,X_i\cap X_j=\emptyset\  (i\neq j)\}$. A function $f:(k+1)^{V}\to \mathbb{R}$ is called $k$-submodular if we have  
\[ f(\boldsymbol{x})+f(\boldsymbol{y})  \geq f(\boldsymbol{x}\sqcap \boldsymbol{y})+f(\boldsymbol{x}\sqcup \boldsymbol{y}) \]
for any $\boldsymbol{x}=(X_1,...,X_k),\ \boldsymbol{y}=(Y_1,...,Y_k) \in (k+1)^{V}$.
Note that
\begin{eqnarray}
\boldsymbol{x}\sqcap \boldsymbol{y}&=&(X_1\cap Y_1,...,X_k\cap Y_k)\ \ and \nonumber \\
\boldsymbol{x}\sqcup \boldsymbol{y}&=&(X_1\cup Y_1\backslash \bigcup_{i\neq 1}(X_i\cup Y_i),...,X_k\cup Y_k\backslash \bigcup_{i\neq k}(X_i\cup Y_i)) \nonumber.
\end{eqnarray}
\end{definition}
It is a submodular function if $k=1$. It is called a bisubmodular function if $k=2$. 

Maximization for $k$-submodular functions is also NP-hard and approximation algorithm have been studied.  An input of the problem is a nonnegative $k$-submodular function. Note that, for any $k$-submodular function $f$ and any $c \in \mathbb{R}$, a function $f'(\boldsymbol{x}):=f(\boldsymbol{x})+c$ is $k$-submodular. An output of the problem is $\boldsymbol{x}=(X_1,...,X_k) \in (k+1)^{V}$. Let an input $k$-submodular function be $f$, a maximizer of $f$ be $\boldsymbol{o}$, and an output of an algorithm be $\boldsymbol{s}$. Then we define the approximation ratio of the algorithm as $f(\boldsymbol{s})/f(\boldsymbol{o})$ for deterministic algorithms, and $\mathbb{E}[f(\boldsymbol{s})]/f(\boldsymbol{o})$ for randomized algorithms.  
For bisubmodular functions, \cite{iwata2013bisubmodular} and \cite{Ward:2016:MKS:2983296.2850419} showed that the algorithm for submodular functions in \cite{buchbinder2015tight} can be extended. \cite{Ward:2016:MKS:2983296.2850419} analyzed an extension for $k$-submodular functions. They showed  a randomized $1/(1+a)$-approximation algorithm with $a=max\{1,\sqrt{(k-1)/4}\}$ and a deterministic $1/3$-approximation  algorithm. 
Now we have a $1/2$-approximation algorithm shown in \cite{iwata2016improved}. In particular, for monotone $k$-submodular functions, \cite{iwata2016improved} gave a $\frac{k}{2k-1}$-approximation algorithm. They also showed any $(\frac{k+1}{2k}+\epsilon)$-approximation algorithm requires exponential time value oracle queries. 

In this paper, we give  a deterministic approximation algorithm for monotone $k$-submodular maximization. It satisfies $\frac{k}{2k-1}$-approximation and runs in polynomial-time. Our algorithm is a derandomized version of algorithm for monotone functions in \cite{iwata2016improved}. We also note the derandomization scheme is from \cite{buchbinder2016deterministic}, used for Double Greedy algorithm.

\section{Preliminary}
Define a partial order $\preceq$ on $(k+1)^{V}$ for $\boldsymbol{x}=(X_1,...,X_k)$ and $ \boldsymbol{y}=(Y_1,...,Y_k)$ as follows: 
\[ \boldsymbol{x} \preceq \boldsymbol{y} \overset{\mathrm{def}}{\Longleftrightarrow} X_i\subseteq Y_i(i=1,...,k). \]
Also, for $\boldsymbol{x}=(X_1,...,X_k) \in (k+1)^{V}$, $e\notin \bigcup_{l=1}^k X_l$, and $i \in \{1,...,k\}$, define
\[ \Delta_{e,i}f(\boldsymbol{x})=f(X_1,...,X_{i-1},X_i\cup\{ e \},X_{i+1},...,X_k)-f(X_1,...,X_k). \]
 
A monotone $k$-submodular function is $k$-submodular and satisfies
$f(\boldsymbol{x})\leq f(\boldsymbol{y})$ for any $\boldsymbol{x}=(X_1,...,X_k)$ and  $\boldsymbol{y}=(Y_1,...,Y_k)$ in $(k+1)^{V}$ with $\boldsymbol{x} \preceq \boldsymbol{y}$.

The property of $k$-submodularity can be written as another form.

\begin{theorem}{\rm (\cite{Ward:2016:MKS:2983296.2850419} THEOREM 7)}
A function $f:(k+1)^{V} \to \mathbb{R}$ is $k$-submodular if and only if $f$ is orthant submodular and pairwise monotone. 

Note that orthant submodularity is to satisfy
\[ \Delta_{e,i} f(\boldsymbol{x}) \geq \Delta_{e,i}f(\boldsymbol{y})\ \  (\boldsymbol{x},\boldsymbol{y}\in (k+1)^{V},\ \boldsymbol{x} \preceq \boldsymbol{y},\ e \notin \bigcup_{l=1}^k Y_l, \ i \in \{1,...,k\} ), \]
and pairwise monotonicity is to satisfy 
\[ \Delta_{e,i}f(\boldsymbol{x}) + \Delta_{e,j}f(\boldsymbol{x}) \geq 0\  (\boldsymbol{x} \in (k+1)^{V},\ e \notin \bigcup_{l=1}^k X_l, \ i,j \in \{ 1,...,k \} \ (i\neq j)). \]
\end{theorem}

To analyze $k$-submodular functions, it is often convenient to identify $(k+1)^V$ as $\{0,1,...,k\}^V$. A $|V|$-dimensional vector $\boldsymbol{x} \in \{0,1,...,k\}^V$ is associated with $(X_1,...,X_n) \in (k+1)^V$ by $X_i = \{ e \in V\mid \boldsymbol{x}(e) = i \}$.

\section{Existing randomized algorithms}
\subsection{Algorithm framework}
In this section, we see the framework to maximize $k$-submodular functions (Algorithm \ref{metakmax} \cite{iwata2016improved}). \cite{iwata2013bisubmodular} and \cite{Ward:2016:MKS:2983296.2850419} used it with specific distributions.\\

\begin{algorithm}
\caption{(\cite{iwata2016improved}\ Algorithm 1)}\label{metakmax}
\begin{algorithmic}
\item{ {\bf Input:} A nonnegative $k$-submodular function $f:\{0,1,...,k\}^{V}\to \mathbb{R}_+$.}
\item{ {\bf Output:} A vector $\boldsymbol{s} \in \{0,1,...,k\}^{V}$.}
\State{$\boldsymbol{s} \gets \boldsymbol{0}\ $.}
\State{Denote the elements of $V$ by $e^{(1)},...,e^{(n)}\ (|V|=n)$.}
\For{$j=1,...,n$}
	\State{Set a probability distribution $p^{(j)}$ over $\{ 1,...,k\}$.}
	\State{Let $\boldsymbol{s}(e^{(j)}) \in \{1,...,k\}$ be chosen randomly with $\mathrm{Pr}[\boldsymbol{s}(e^{(j)})=i]=p_i^{(j)}$.}
\EndFor \\
\Return{$\boldsymbol{s}$}
\end{algorithmic}
\end{algorithm}

Algorithm \ref{metakmax} is not only used for monotone functions. However, in this paper, we only use it for monotone functions.

Now we define some variables to see Algorithm \ref{metakmax}. Let $\boldsymbol{o}$ be an optimal solution, and we write $\boldsymbol{s}^{(j)}$ as $\boldsymbol{s}$ at $j$-th iteration. 
Let other variables be as follows: 
\begin{eqnarray}
\boldsymbol{o}^{(j)}=(\boldsymbol{o}\sqcup\boldsymbol{s}^{(j)})\sqcup \boldsymbol{s}^{(j)} &,\  & \boldsymbol{t}^{(j-1)}(e)=
\begin{cases} \boldsymbol{o}^{(j)}(e) &(e\neq e^{(j)})  \\ 0 & (e= e^{(j)}) \end{cases}\nonumber \\
y_i^{(j)}=\Delta_{e^{(j)},i}f(\boldsymbol{s}^{(j-1)}) &,\ &
a_i^{(j)}=\Delta_{e^{(j)},i}f(\boldsymbol{t}^{(j-1)}) \nonumber
\end{eqnarray}

Algorithm \ref{metakmax} satisfies following lemma.

\begin{lemma}\label{ksub_c} {\rm (\cite{iwata2016improved}\ LEMMA 2.1.)}\\
Let $c\in \mathbb{R}_+$. Conditioning on $\boldsymbol{s}^{(j-1)}$, suppose that
\begin{equation}
\sum_{i=1}^k(a_{i^*}^{(j)}-a_{i}^{(j)})p_i^{(j)}\leq c \sum_{i=1}^k(y_{i}^{(j)}p_i^{(j)}) \nonumber
\end{equation}
holds for each $j$ with $1\leq j \leq n$, where $i^*=\boldsymbol{o}(e^{(j)})$. Then $\mathbb{E}[f(\boldsymbol{s})]\geq \frac{1}{1+c}f(\boldsymbol{o})$.
\end{lemma}

\subsection{A randomized algorithm for monotone functions}
In \cite{iwata2016improved}, a randomized $\frac{k}{2k-1}$-approximation algorithm for monotone $k$-submodular functions (Algorithm \ref{kmax}) is shown.  

\begin{algorithm}
\caption{(\cite{iwata2016improved}\ Algorithm 3)}\label{kmax}
\begin{algorithmic}
\item{ {\bf Input:} A monotone $k$-submodular function $f:\{0,1,...,k\}^{V}\to \mathbb{R}_+$.}
\item{ {\bf Output:} A vector $\boldsymbol{s} \in \{0,1,...,k\}^{V}$.}
\State{$\boldsymbol{s} \gets \boldsymbol{0},t\gets k-1$.}
\State{Denote the elements of $V$ by $e^{(1)},...,e^{(n)}\ (|V|=n)$.}
\For{$j=1,...,n$}
	\State{$y_i^{(j)} \gets \Delta_{e^{(j)},i}f(\boldsymbol{s})\ \ (1\leq i\leq k)$.}
	\State{$\beta \gets \sum^k_{i=1}(y_i^{(j)})^t$.}
	\If{$\beta \neq 0$}{\ \ $p_i^{(j)} \gets (y_i^{(j)})^t/\beta \ \ (1\leq i\leq k)$.}
	\Else 
	\State{$p_1^{(j)}=1,p_i^{(j)}=0\ (i=2,...,k)$.}
	\EndIf
	\State{Let$\boldsymbol{s}(e^{(j)}) \in \{1,...,k\}$ be chosen randomly with $\mathrm{Pr}[\boldsymbol{s}(e^{(j)})=i]=p_i^{(j)}$.}
\EndFor \\
\Return{$\boldsymbol{s}$}
\end{algorithmic}
\end{algorithm}

Algorithm \ref{kmax} runs in polynomial time. The approximation ratio of Algorithm \ref{kmax}  satisfies the theorem below.

\begin{theorem}{\rm (\cite{iwata2016improved}\ THEOREM 2.2.)}
Let $\boldsymbol{o}$ be a maximizer of a monotone $k$-submodular function $f$ and let $\boldsymbol{s}$ be the output of Algorithm \ref{kmax}. Then $\mathbb{E}[f(\boldsymbol{s})]\geq \frac{k}{2k-1}f(\boldsymbol{o})$.
\end{theorem}

In the proof of this theorem (see \cite{iwata2016improved}), the inequality of Lemma \ref{ksub_c} is proved with $c=1-\frac{1}{k}$. We get $a_{i} \geq 0\ (\forall i\in \{1,...,k\})$ from monotonicity, and $a_{i} \leq y_i \ (\forall i\in \{1,...,k\})$ from orthant submodularity. Hence, the inequality
\begin{equation}
\sum_{i\neq i^*}(y_{i^*}^{(j)}p_i^{(j)})\leq \left(1-\frac{1}{k}\right) \sum_{i=1}^k(y_{i}^{(j)}p_i^{(j)}) 
\end{equation}
is used. The inequality of Lemma \ref{ksub_c} is satisfied when the inequality (1) is valid. 

\section{Deterministic algorithm}
In this section, we give a polynomial-time deterministic algorithm for maximizing monotone $k$-submodular functions.  Our algorithm is Algorithm \ref{Dkmax}. Algorithm \ref{Dkmax} is a derandomized version of  Algorithm \ref{kmax}. We note the derandomization scheme of  this algorithm is from \cite{buchbinder2016deterministic}. \\

\begin{algorithm}
\caption{A deterministic algorithm}\label{Dkmax}
\begin{algorithmic}
\item{ {\bf Input:} A monotone $k$-submodular function $f:\{0,1,...,k\}^{V}\to \mathbb{R}_+$.}
\item{ {\bf Output:} A vector $\boldsymbol{s} \in \{0,1,...,k\}^{V}$.}
\State{$\mathcal{D}_{0}\gets(1,\ \boldsymbol{0}),\ (\mathcal{D}=\{ (p,\boldsymbol{s}) \mid \boldsymbol{s}\in (k+1)^{V},\ 0\leq p\leq 1\}\ (\sum_{\boldsymbol{s} \in \mathcal{D}} p =1))$.}
\State{Denote the elements of $V$ by $e^{(1)},...,e^{(n)}\ (|V|=n)$.}
\For{$j=1,...,n$}
	\State{$y_i(\boldsymbol{s}) \gets \Delta_{e^{(j)},i}f(\boldsymbol{s})\ (\forall \boldsymbol{s} \in \mathrm{supp}(\mathcal{D}_{j-1}) ,\  i \in \{1,...,k\})$.}
	\State{Find an extreme point solution $(p_{i,\boldsymbol{s}})_{i=1,\ldots,k,\ \boldsymbol{s} \in \mathrm{supp}(\mathcal{D}_{j-1})}$of the following linear formulation:
\begin{eqnarray}
  \left(1-\frac{1}{k}\right)  \mathbb{E}_{\boldsymbol{s}\sim \mathcal{D}_{j-1}}\left[\sum_{i=1}^k p_{i,\boldsymbol{s}}y_{i}(\boldsymbol{s}) \right]&\geq &\mathbb{E}_{\boldsymbol{s}\sim \mathcal{D}_{j-1}}\left[(1-p_{l,\boldsymbol{s}})y_{l}(\boldsymbol{s})\right] \\
	&&\ \ \ \ (l\in \{1,...,k\}) \nonumber \\
   \sum_{i=1}^k p_{i,\boldsymbol{s}}&=&1\ \ ( \forall \boldsymbol{s} \in \mathrm{supp}(\mathcal{D}_{j-1})) \\
	p_{i,\boldsymbol{s}}&\geq&0\ \ ( \forall \boldsymbol{s} \in \mathrm{supp}(\mathcal{D}_{j-1}),\ i \in \{1,...,k\}).  
\end{eqnarray}}
\State{Construct a new distribution $\mathcal{D}_j$:
\begin{eqnarray}
\mathcal{D}_j\gets  \bigcup_{i=1}^k \{ (p_{i,\boldsymbol{s}} \cdot \mathrm{Pr}_{\mathcal{D}_{j-1}}[\boldsymbol{s}],\ \boldsymbol{s}_{e^{(j)},i}) \mid \boldsymbol{s} \in \mathrm{supp}(\mathcal{D}_{j-1}),\ p_{i,\boldsymbol{s}}>0 \} 
\end{eqnarray}
\[ \left(\boldsymbol{s}_{e^{(j)},i}(e)=
\begin{cases} \boldsymbol{s}(e) &(e\neq e^{(j)})  \\ i & (e= e^{(j)}) \end{cases}\right).\]
}
\EndFor \\
\Return{$\argmax_{\boldsymbol{s} \in \mathrm{supp}(\mathcal{D}_{n})} \{f(\boldsymbol{s})\}$}
\end{algorithmic}
\end{algorithm}

In the algorithm, we construct a distribution $\mathcal{D}$ which satisfies $\mathbb{E}_{\boldsymbol{s} \sim \mathcal{D}}[f(\boldsymbol{s})]\geq \frac{k}{2k-1}f(\boldsymbol{o})$. Then the algorithm outputs the best solution in $\mathrm{supp}(\mathcal{D}):=\{ \boldsymbol{s} \mid (p,\boldsymbol{s}) \in \mathcal{D} \}$. We can see the right hand side of (2) in Algorithm \ref{Dkmax} is the expected value of the left hand side  of (1) for $\boldsymbol{s} \sim \mathcal{D}_{j-1}$ with $i^* = l$. it is because $\sum_{i\neq l} p_{i, \boldsymbol{s}}y_l(\boldsymbol{s}) = (1-p_{l, \boldsymbol{s}})y_l(\boldsymbol{s})$. Also the left hand side of (2) is the expected value of the right hand side  of (1) with $c=1-1/k$. From (3) and (4), $\mathcal{D}_{j}$ in (5) is constructed as a distribution.\\

Algorithm \ref{Dkmax} achieves the same approximation ratio as Algorithm \ref{kmax}.
\begin{theorem}
Let $\boldsymbol{o}$ be a maximizer of  a monotone nonnegative k-submodular function $f$ and let $\boldsymbol{z}$ be the output of Algorithm \ref{Dkmax}. Then $f(\boldsymbol{z})\geq \frac{k}{2k-1}f(\boldsymbol{o})$.
\end{theorem}
\begin{proof}
We consider the $j$-th iteration. From (5), we get
\begin{eqnarray}
  \mathbb{E}_{\boldsymbol{s}\sim \mathcal{D}_{j-1}}\left[\sum_{i=1}^k p_{i,\boldsymbol{s}}y_{i}(\boldsymbol{s}) \right]&=&\mathbb{E}_{\boldsymbol{s}\sim \mathcal{D}_{j-1}}\left[\sum_{i=1}^k p_{i,\boldsymbol{s}}(f(\boldsymbol{s}_{e^{(j)},i})-f(\boldsymbol{s}) )\right] \nonumber  \\
&=&\mathbb{E}_{\boldsymbol{s}\sim \mathcal{D}_{j-1}}\left[\sum_{i=1}^k p_{i,\boldsymbol{s}}f(\boldsymbol{s}_{e^{(j)},i})-f(\boldsymbol{s}) \right]  \nonumber  \\
&=&\mathbb{E}_{\boldsymbol{s}'\sim \mathcal{D}_{j}}\left[ f(\boldsymbol{s}')\right]-\mathbb{E}_{\boldsymbol{s}\sim \mathcal{D}_{j-1}}\left[f(\boldsymbol{s}) \right].   
\end{eqnarray}
Now, we consider $\boldsymbol{o}[\boldsymbol{s}]:=(\boldsymbol{o} \sqcup \boldsymbol{s})\sqcup \boldsymbol{s}$.
Define the variables as follows:
\begin{eqnarray}
\boldsymbol{r}(e)&=&
\begin{cases} \boldsymbol{o}[\boldsymbol{s}](e) &(e\neq e^{(j)})  \\ 0 & (e= e^{(j)}) \end{cases} \nonumber \\
a_i(\boldsymbol{s})&=&\Delta_{e^{(j)},i}f(\boldsymbol{r}) \nonumber
\end{eqnarray}
Then we have
\begin{eqnarray}
f(\boldsymbol{o}[\boldsymbol{s}])-f(\boldsymbol{o}[\boldsymbol{s}_{e^{(j)},i}])&=&a_{i^*}(\boldsymbol{s})-a_{i}(\boldsymbol{s})\ \ (i^*=\boldsymbol{o}(e^{(j)}))
\end{eqnarray}
From monotonicity and orthant submodularity of $f$, we have
\begin{eqnarray}
a_{i^*}(\boldsymbol{s})-a_{i}(\boldsymbol{s})\leq y_{i^*}(\boldsymbol{s}).
\end{eqnarray}
From (7) and (8), we get
\begin{eqnarray}
\mathbb{E}_{\boldsymbol{s}\sim \mathcal{D}_{j-1}}\left[ f(\boldsymbol{o}[\boldsymbol{s}])\right]-\mathbb{E}_{\boldsymbol{s}'\sim \mathcal{D}_{j}}\left[f(\boldsymbol{o}[\boldsymbol{s}']) \right]&=&\mathbb{E}_{\boldsymbol{s}\sim \mathcal{D}_{j-1}}\left[\sum_{i=1}^k p_{i,\boldsymbol{s}}f(\boldsymbol{o}[\boldsymbol{s}])-f(\boldsymbol{o}[\boldsymbol{s}_{e^{(j)},i}]) \right] \nonumber \\
&=&\mathbb{E}_{\boldsymbol{s}\sim \mathcal{D}_{j-1}}\left[\sum_{i=1}^k p_{i,\boldsymbol{s}}\left( f(\boldsymbol{o}[\boldsymbol{s}])-f(\boldsymbol{o}[\boldsymbol{s}_{e^{(j)},i}])\right)  \right] \nonumber \\
&=&\mathbb{E}_{\boldsymbol{s}\sim \mathcal{D}_{j-1}}\left[\sum_{i\neq i^*} p_{i,\boldsymbol{s}}\left( a_{i^*}(\boldsymbol{s})-a_{i}(\boldsymbol{s}) \right)  \right] \nonumber \\
&\leq &\mathbb{E}_{\boldsymbol{s}\sim \mathcal{D}_{j-1}}\left[\sum_{i\neq i^*} p_{i,\boldsymbol{s}}\left( y_{i^*}(\boldsymbol{s}) \right)  \right] \nonumber \\
&=&\mathbb{E}_{\boldsymbol{s}\sim \mathcal{D}_{j-1}}\left[(1-p_{i^*,\boldsymbol{s}})\left( y_{i^*}(\boldsymbol{s}) \right)  \right] .
\end{eqnarray}
$p_{i,\boldsymbol{s}}$ satisfies (2) for all $l \in \{1,2,...,k\}$.  Hence we obtain 
\begin{equation}
\left(1-\frac{1}{k}\right)\left( \mathbb{E}_{\boldsymbol{s}'\sim \mathcal{D}_{j}}\left[ f(\boldsymbol{s}')\right]-\mathbb{E}_{\boldsymbol{s}\sim \mathcal{D}_{j-1}}\left[f(\boldsymbol{s}) \right]\right) \geq \mathbb{E}_{\boldsymbol{s}\sim \mathcal{D}_{j-1}}\left[ f(\boldsymbol{o}[\boldsymbol{s}])\right]-\mathbb{E}_{\boldsymbol{s}'\sim \mathcal{D}_{j}}\left[f(\boldsymbol{o}[\boldsymbol{s}']) \right]
\end{equation}
from (6) and (9). By the summation of (10), we get
\begin{equation}
\left(1-\frac{1}{k}\right)\left(\mathbb{E}_{\boldsymbol{s}'\sim \mathcal{D}_{n}}\left[ f(\boldsymbol{s}')\right]-\mathbb{E}_{\boldsymbol{s}\sim \mathcal{D}_{0}}\left[f(\boldsymbol{s}) \right]\right) \geq \mathbb{E}_{\boldsymbol{s}\sim \mathcal{D}_{0}}\left[ f(\boldsymbol{o}[\boldsymbol{s}])\right]-\mathbb{E}_{\boldsymbol{s}'\sim \mathcal{D}_{n}}\left[f(\boldsymbol{o}[\boldsymbol{s}']) \right].
\end{equation}
Note that $\boldsymbol{o}[\boldsymbol{s}']=\boldsymbol{s}'$ for $\boldsymbol{s}' \in \mathrm{supp}(\mathcal{D}_n)$, and $\boldsymbol{o}[\boldsymbol{s}]=\boldsymbol{o}$ for $\boldsymbol{s} \in \mathrm{supp}(\mathcal{D}_0)$. Now we have
\begin{eqnarray}
f(\boldsymbol{o})&\leq& \left(2-\frac{1}{k}\right) \mathbb{E}_{\boldsymbol{s}'\sim \mathcal{D}_{n}}\left[ f(\boldsymbol{s}')\right]-\left(1-\frac{1}{k}\right)f(\boldsymbol{0})\nonumber \\
&\leq& \left(2-\frac{1}{k}\right) \mathbb{E}_{\boldsymbol{s}'\sim \mathcal{D}_{n}}\left[ f(\boldsymbol{s}')\right]\nonumber \\
&\leq&\left(2-\frac{1}{k}\right) \max_{\boldsymbol{s}' \in \mathrm{supp}(\mathcal{D}_{n})} \{f(\boldsymbol{s}')\} \nonumber 
\end{eqnarray}
\end{proof}

The algorithm performs a polynomial number of value oracle queries.
 
\begin{theorem}
Algorithm \ref{Dkmax} returns a solution after $\mathrm{O}(n^2k^2)$ value oracle queries.
\end{theorem}
\begin{proof}
Algorithm \ref{Dkmax} uses the value oracle to caluculate $y_i(\boldsymbol{s})$. At $j$-th iteration, the number of  $y_i(\boldsymbol{s})$ is $k|\mathcal{D}_{j-1}|$. From (5), $|\mathcal{D}_{j}|$ equals the number of $p_{i,\boldsymbol{s}}\neq 0$. Then we have to consider $p_{i,\boldsymbol{s}}\neq 0$ at $j$-th iteration.

By the definition, $(p_{i,\boldsymbol{s}})_{i=1,\ldots,k,\ \boldsymbol{s} \in \mathrm{supp}(\mathcal{D}_{j-1})}$ is an extreme point solution of (2), (3), and (4). Note that, we can get a solution by setting  $(p_{i,\boldsymbol{s}})$ as the distribution of Algorithm \ref{kmax} for each  $\boldsymbol{s} \in \mathrm{supp}(\mathcal{D}_{j-1})$. We can also see the feasible region of (2), (3), and (4) is bounded. Then some extreme point solution exists.

Let $|\mathcal{D}_{j-1}|=m$. By $(p_{i,\boldsymbol{s}})_{i=1,\ldots,k,\ \boldsymbol{s} \in \mathrm{supp}(\mathcal{D}_{j-1})}\in \mathbb{R}^{km}$ and $k$ equalities of (3), $km-k$ inequalities are tight at any extreme point solution. (2) have $m$ inequalities and (4) have $km$ inequalities. Then, at least  $km-k-m$ inequalities of (4) are tight. Hence, the number of $p_{i,\boldsymbol{s}}\neq 0$ is at most $m+k$.

Now we have $|\mathcal{D}_{j}|\leq |\mathcal{D}_{j-1}|+k$. We can also see $|\mathcal{D}_{j}|\leq jk+1$. Then the number of value oracle queries is
\[ \sum_{j=1}^n k|\mathcal{D}_{j-1}| \leq \sum_{j=1}^n k (jk +1). \]
\end{proof}

In our algorithm, we have to search for an extreme point solution. We can do it by solving LP for some objective function. If we use LP for our algorithm, it is polynomial-time not only for the number of queries but also for the number of operations. 
The simplex method is not proved to be a polynomial-time method. However, it is practical. Our algorithm needs only an extreme point solution, then if we get a basic solution, it is enough. So we can use the first phase of  two-phase simplex method to find an extreme point solution.

\section{Conclusion}
We showed a derandomized algorithm for monotone $k$-submodular maximization. It is $\frac{k}{2k-1}$-approximation and polynomial-time algorithm.  

One of open problems is a faster method for finding an extreme point solution of the linear formulation. For submodular functions, \cite{buchbinder2016deterministic} showed greedy methods are effective. It is because their formulation is the form of fractional knapsack problem.  Our formulation is similar to theirs, and ours can be seen as the form of an LP relaxation of multidimensional knapsack problem. However, faster methods are not given than general LP solutions. The number of constraints in our formulation depends on $k$ and the number of iterations. It is therefore difficult to find an extreme point faster.

Constructing a deterministic algorithm for nonmonotone functions is also an important open problem. For nonmonotone functions, we have pairwise monotonicity instead of monotonicity. In such a situation, for some $i$, $a_i$ can be negative. However, if $y_j > 0$ for all $j$, we can't find such $i$. Then, if we try to use the same derandomizing method, the number of constraints in the linear formulation and the size of $\mathcal{D}$ will be exponential. So algorithm can't finish in polynomial-time.

\bibliographystyle{plain.bst}
\bibliography{suririnko_yoko_ksub_ref.bib}

\end{document}